\documentclass[journal]{IEEEtran}

\ifCLASSINFOpdf
\else
   \usepackage[dvips]{graphicx}
\fi
\usepackage{url}
\usepackage{color}
\usepackage{graphicx}

\usepackage{amsmath,amssymb,amsthm}
\usepackage{nicefrac}       
\usepackage[]{algorithm2e}  
\usepackage{etoolbox}       

\newcommand{\N}{\mathbb{N}}
\newcommand{\R}{\mathbb{R}}

\newcommand{\Lmat}{\mathbf{L}}

\newtheorem{theorem}{Theorem}[section]
\newtheorem{corollary}[theorem]{Corollary}
\newtheorem{lemma}[theorem]{Lemma}

\newbool{final}
\setbool{final}{false}
\ifbool{final}{
    \newcommand{\pg}[1]{}
    \newcommand{\ab}[1]{}
    \newcommand{\RG}[1]{}
    \newcommand{\sm}[1]{}
    \newcommand{\ms}[1]{}
}{
    \newcommand{\pg}[1]{\color{magenta}[#1]\color{black}~}
    \newcommand{\ab}[1]{\color{cyan}[#1]   \color{black}~}
    \newcommand{\RG}[1]{\color{red}[#1]    \color{black}~}
    \newcommand{\sm}[1]{\color{blue}[#1]   \color{black}~}
    \newcommand{\ms}[1]{\color{green}[#1]   \color{black}~}
}

\begin{document}

\title{
Fast Multiscale Diffusion on Graphs}

\author{
    Sibylle Marcotte,
    Amélie Barbe,
    Rémi Gribonval, 
    Titouan Vayer, \\
    Marc Sebban,
    Pierre Borgnat, 
    Paulo Gon\c{c}alves 

\thanks{
    Work supported by the ACADEMICS grant of the IDEXLYON, project of the Universit\'e de Lyon, PIA operated by ANR-16-IDEX-0005.
    }
\thanks{
S. Marcotte is with ENS Rennes. France (email: sibylle.marcotte@ens-rennes.fr). 
A. Barbe, R. Gribonval, T.Vayer and P. Gonçalves are with  Universit\'e de Lyon, Inria, CNRS, ENSL, LIP, Lyon; P. Borgnat is with Universit\'e de Lyon, ENSL, CNRS, Laboratoire de Physique, Lyon, France (emails: first.last@ens-lyon.fr).
M. Sebban is with Univ Lyon, UJM-Saint-Etienne, CNRS, Institut d Optique Graduate School, Laboratoire Hubert Curien, Saint-Etienne, France (email: marc.sebban@univ-st-etienne.fr).
}
}

\markboth{}
{Shell \MakeLowercase{\textit{et al.}}: Bare Demo of IEEEtran.cls for IEEE Journals}
\maketitle

\begin{abstract}
Diffusing a graph signal at multiple scales requires computing the action of the exponential of several multiples of the Laplacian matrix. We tighten a bound on the approximation error of truncated Chebyshev polynomial approximations of the exponential, hence significantly improving a priori estimates of the polynomial order for a prescribed error. We further exploit properties of these approximations to factorize the computation of the action of the diffusion operator over multiple scales, thus reducing drastically its computational cost.
\end{abstract}

\begin{IEEEkeywords}
    Approximate computing,
    Chebyshev approximation,
    Computational efficiency,
    Estimation error,
    Polynomials
\end{IEEEkeywords}

\IEEEpeerreviewmaketitle

\vspace{-0.5cm}

\section{introduction}
\label{S:introduction}

\IEEEPARstart{T}{he} matrix exponential operator has applications in numerous domains, ranging from time integration of Ordinary Differential Equations~\cite{mattheij2005partial} or network analysis~\cite{de2019analysis} to various simulation problems (like power grids~\cite{zhuang2013power} or nuclear reactions~\cite{pusa2010computing}) or machine learning~\cite{barbe2020graph}.
In graph signal processing, it appears in the diffusion process of a graph signal -- an analog on graphs of Gaussian low-pass filtering.

Given a graph \(\mathcal{G}\) and its combinatorial Laplacian matrix \(\Lmat\), let \(x\) be a signal  on this graph (a vector containing a value at each node), the \emph{diffusion of \(x\) in \(\mathcal{G}\)} is defined by the equation \(\frac{dw}{d\tau} = -\mathbf{L} \cdot w\) with \(w(0) = x\)~\cite{chung1997spectral}.
It admits a closed-form solution \(w(\tau) = \exp(-\tau \mathbf{L}) x\) involving the \emph{heat kernel} \(\tau \rightarrow \exp(-\tau \mathbf{L})\), which features the matrix exponential.

Applying the exponential of a matrix \(\mathbf{M} \in \mathbb{R}^{n \times n}\) to a vector \(x \in \mathbb{R}^n\) can be achieved by computing the matrix \(\mathbf{B} = \exp(\mathbf{M})\) to compute then the matrix-vector product \(\mathbf{B}x\).
However, this becomes quickly computationally prohibitive in high dimension, as storing and computing \(\mathbf{B}\), as well as the matrix-vector product \(\mathbf{B}x\), have cost at least quadratic in \(n\).
Moreover, multiscale graph representations such as graph wavelets \cite{DBLP:journals/adcm/MehraSL21}, graph-based machine learning methods \cite{barbe2020graph}, rely on graph diffusion at different scales, thus implying applications of the matrix exponential of various multiples of the graph Laplacian. 

To speedup such repeated computations 
one can use a well-known technique based on approximations of the (scalar) exponential function using Chebyshev polynomials.
We build on the fact that polynomial approximations \cite{DBLP:journals/siammax/PopolizioS08} can significantly reduce the computational burden of approximating \(\exp(\mathbf{M})x\) with good precision when \(\mathbf{M}=-\tau\mathbf{L}\) where $\mathbf{L}$ is  sparse positive semi-definite (PSD);
this is often the case when \(\mathbf{L}\) is the Laplacian of a graph when each node is connected to a limited number of neighbors.
The principle is to approximate the exponential as a low-degree polynomial in \(\mathbf{M}\), \(\exp(\mathbf{M}) \approx p(\mathbf{M}) := \sum_{k=0}^K a_k \mathbf{M}^k\).
Several methods exist, some requiring the explicit computation of coefficients associated with a particular choice of polynomial basis, others, including Krylov-based techniques, not requiring explicit evaluation of the coefficients but relying on an iterative determination \cite{DBLP:journals/jcam/BotchevK20} of the polynomial approximation  on the subspace spanned by \(\left\{ x, \mathbf{M}x, \cdots, \mathbf{M}^K x \right\}\).

Our contribution is twofold.
First, we devise a new bound on the 
approximation error of truncated Chebyshev expansions of the exponential, that improves upon existing works~\cite{DRUSKIN1989112,DBLP:journals/nla/BergamaschiV00,mason2002chebyshev}.
This avoids unnecessary computations by determining a small truncation order \(K\) to achieve a prescribed error. Second, we propose to compute \(\exp(-\tau\mathbf{L})\) at different scales \(\tau \in \mathbb{R}\) faster, by reusing the calculations of the action of Chebyshev polynomials on $x$ 
and combining them with adapted coefficients for each scale $\tau$. 
This is particularly efficient for multiscale problems with arbitrary values of \(\tau\), unlike \cite{DBLP:journals/siamsc/Al-MohyH11} which is limited to linear spacing. 

The rest of this document is organized as follows.
In Section~\ref{S:contribution} we describe univariate function approximation with Chebyshev polynomials, and detail the approximation of scaled univariate exponential functions with new bounds on the coefficients (Corollary~\ref{L:coeffdkbk}). This is used in 
Section~\ref{S:properties} to approximate matrix exponentials with controlled complexity and controlled error (Lemma~\ref{L:uniformbound1}), leading to our new error bounds~\eqref{E:Bound1},~\eqref{E:Bound2},~\eqref{E:Bound3}. Section~\ref{S:experiments} is dedicated to an experimental validation, with a comparison to the state-of-the-art bounds of \cite{DBLP:journals/nla/BergamaschiV00}, and an illustration on multiscale diffusion.

\section{Chebyshev approximation of the exponential}
\label{S:contribution}

The Chebyshev polynomials of the first kind are characterized by the identity \(T_k(\cos(\theta)) = \cos(k\theta)\).
They can be computed as $T_0(t)=1$, $T_1(t) = t$ and using the following recurrence relation:
\begin{align}\label{E:recur_cheb}
    T_{k+2}(t)&= 2tT_{k+1}(t)-T_k(t).
\end{align}
%
The \emph{Chebyshev series decomposition} of a function \(f : [-1,1] \mapsto \R\) is:
    $f(t) = \frac{c_0}{2} + \sum_{k \geq 1} c_k \cdot T_k(t)$,
where the \emph{Chebyshev coefficients} are:
\begin{equation}
   \label{E:def_chebycheb_coefficients}
    c_k = \frac{2}{\pi} \int_0^{\pi}
            \cos(k\theta) \cdot f(\cos(\theta))
        \mathrm{d}\theta.
\end{equation}
Truncating this series yields an approximation of \(f\).
For theoretical aspects of the approximation by Chebyshev polynomials (and other polynomial basis) we refer the reader to~\cite{phillips2003interpolation}.

\subsection{Chebyshev series of the exponential}
\label{S:genericsetting}

We focus on approximating the univariate transfer function \(h_{\tau}: \lambda \in [0, 2] \mapsto \exp(-\tau \lambda)\), which will be useful to obtain low-degree polynomial approximations of the matrix exponential $\exp(-\tau \mathbf{L})$ for positive semi-definite matrices whose largest eigenvalue satisfies $\lambda_{\max}= 2$ (see Section~\ref{S:properties}). 

Using a change of variable:

$t = (\lambda-1) \in[-1,1]$, $\tilde{h}_{\tau}(t) = h_{\tau}(t+1)$

and the Chebyshev series 
of \(f := \tilde{h}_\tau\) yields:
\begin{align}
    \tilde{h}_{\tau}(t)
        &= \frac{1}{2} c_0(\tau) + \sum_{k=1}^{\infty} c_k(\tau) T_k(t),\notag \\
    \label{E:diff_approx_coeff}
    c_k(\tau) &=
        \frac{2}{\pi} \int_0^{\pi}
            \cos(k\theta) \exp(-\tau (\cos(\theta)+1))
        \mathrm{d}\theta.
\end{align}
This leads to the following expression for \(h_{\tau}\):
\begin{equation}
    \label{E:diff_approx}
    h_{\tau}(\lambda) = \frac{1}{2} c_0(\tau) + \sum_{k=1}^{\infty} c_k(\tau) \tilde{T}_k(\lambda),
\end{equation}
where for any \(k \in \N\):
    $\tilde{T}_k(\lambda)= T_k\left( \lambda - 1 \right)$.

Truncating the series~\eqref{E:diff_approx} to order \(K\) yields a polynomial approximation of \(h_\tau\) of degree $K$ whose quality can be controlled, leading to a control of the error in approximating the action of \(\exp(-\tau \mathbf{L})\) as studied in Section~\ref{S:properties}.
First we focus on how to evaluate the coefficients \(c_k\) defined in Equation~\eqref{E:diff_approx_coeff}.

\subsection{Chebyshev coefficients of the exponential operator}

Evaluating numerically the coefficients using the integral formulation ~\eqref{E:diff_approx_coeff} would be computationally costly, fortunately they are expressed using Bessel functions~\cite{abramowitz1988handbook}:
\begin{equation}
    \label{E:coeff_as_bessel_fun}
    c_k(\tau) = 2I_k(\tau) \cdot \exp(-\tau) = 2 \cdot {Ie}_k(-\tau),
\end{equation}
with \(I_k(\cdot)\) the modified Bessel function of the first kind and \({Ie}_k(\cdot)\) the exponentially scaled modified Bessel function of the first kind.

The following lemma applied to \(f = \tilde{h}_\tau\) yields another expression of the coefficients \eqref{E:diff_approx_coeff}, which will be used to bound the error of the truncated Chebyshev expansion.
\begin{lemma}[\cite{phillips2003interpolation}, Equation~2.91]
    \label{L:philips_coeff}
    Let \(f\) be a function expressed as an infinite power series \(f(t)= \sum_{i=0}^{\infty} a_i t^{i}\) and assume that this series is uniformly convergent on \([-1, 1]\).
    Then, we can express the Chebyshev coefficients of \(f\) by:
    \begin{equation}
        c_k = \frac{1}{2^{k-1}} \sum_{i=0}^{\infty} \frac{1}{2^{2i}} \binom{k+2i}{i} a_{k+2i}.
    \end{equation}
\end{lemma}

\begin{corollary}
    \label{L:coeffdkbk}
    Consider \(\tilde{h}_\tau(t):=\exp(-\tau(t+1))\), \(t \in [-1,1]\).
    The coefficients of its Chebyshev expansion satisfy:
    \begin{align}
        c_k &=
            (-1)^k d_k \bar{c}_k
            \label{E:ckdk}\\
        \bar{c}_k &=
            2 \left(\nicefrac{\tau}{2}\right)^k \exp(-\tau) (k!)^{-1}
            \label{E:cbark}\\
        d_k &=
            \sum_{i=0}^\infty \left(\nicefrac{\tau}{2}\right)^{2i} \frac{k!}{i!(k+i)!}.
            \label{E:dk}    \end{align}
    Moreover we have:
    \begin{equation}
        \label{E:dk_bound}
        1 \leq d_k \leq \min\left(
            \exp\left(\frac{
                (\nicefrac{\tau}{2})^2}{
                k+1}\right)
            ,\cosh(\tau)
        \right).
    \end{equation}
\end{corollary}

\begin{proof}
Denoting \(C = \tau/2\), we expand \(f(t)=\tilde{h}_\tau(t) = \exp(-2C(t+1)) = \exp(-2C)\exp(-2Ct)\) into a power series:
\begin{equation*}
    f(t) = \sum_{i=0}^{\infty} \exp(- 2C) \frac{(-2C)^{i}}{i!} t^{i}.
\end{equation*}
Using Lemma~\ref{L:philips_coeff}, we obtain for each \(k \in \N\):
\begin{equation*}
    c_k = (-1)^{k} C^{k} 2\exp(-2C) \sum_{i=0}^{\infty} C^{2i} \frac{1}{i!(k+i)!} = (-1)^k \bar{c}_k d_k.
\end{equation*}

For any integers \(k,i\) we have \(k!/(k+i)! \leq \min(1/i!, 1/(k+1)^i)\) and  \(1/(i!)^2= \binom{2i}{i}/(2i)! \leq 2^{2i}/(2i)!\) hence
\begin{align*}
    d_k
        &= \sum_{i=0}^{\infty} \frac{C^{2i}}{i!} \frac{k!}{(k+i)!}\\
        &\leq \min\left(
            \sum_{i=0}^{\infty} \frac{C^{2i}}{i!} \frac{1}{(k+1)^i},
            \sum_{i=0}^{\infty} \frac{C^{2i}}{i!i!}
        \right)\\
        &\leq \min\left(
            \exp\left(C^2/(k+1)\right),
            \sum_{i=0}^{\infty} \frac{C^{2i}2^{2i}}{(2i)!}
        \right)\\
        &= \min\left(
            \exp\left(C^2/(k+1)\right),
            \cosh(2C)
        \right).\qedhere
\end{align*}
\end{proof}

\section{Approximation of the matrix exponential}
\label{S:properties}

The extension of a univariate function \(f: \mathbb{R} \to \mathbb{R}\) to symmetric matrices \(\mathbf{L} \in \mathbb{R}^{n \times n}\) exploits the eigen-decomposition \(\mathbf{L} = \mathbf{U}\boldsymbol{\Lambda}\mathbf{U}^\top\), where \(\boldsymbol{\Lambda} = \mathtt{diag}(\lambda_i)_{1\leq i \leq n}\), to define the action of \(f\) as \(f(\mathbf{L}) := \mathbf{U}\mathtt{diag}(f(\lambda_i)) \mathbf{U}^\top\).
When \(f(t) = t^k\) for some integer \(\)k, this yields \(f(\mathbf{L}) = \mathbf{L}^k\), hence the definition matches with the intuition when \(f\) is polynomial or analytic.

The exponential of a matrix could be computed by taking the exponential of the eigenvalues, but 
diagonalizing the matrix would be computationally prohibitive. 
However computing a matrix such as \(\exp(-\tau\mathbf{L})\) is rarely required, as one rather needs to compute its \emph{action} on a given vector. 
This enables faster methods, notably using polynomial approximations: given a square symmetric matrix \(\mathbf{L}\) and a univariate function \(f\), a suitable univariate polynomial \(p\) is used to approximate \(f(\mathbf{L})\) with \(p(\mathbf{L})\).
Such a polynomial can depend on both \(f\) and \(\mathbf{L}\).
When the function \(f\) admits a Taylor expansion, a natural choice for \(p\) is a truncated version of the Taylor series~\cite{DBLP:journals/siamsc/Al-MohyH11}.
Other polynomial bases can be used, such as the Padé polynomials, or in our case, the Chebyshev polynomials~\cite{DBLP:journals/nla/BergamaschiV00,DBLP:conf/dcoss/ShumanVF11} (see~\cite{DBLP:journals/actanum/HighamA10} for a survey), leading to approximation errors that decay exponentially with the polynomial order $K$.

\subsection{Chebyshev approximation of the matrix exponential}
Consider \(\mathbf{L}\) any \emph{PSD matrix} of largest eigenvalue \(\lambda_{\max} = 2\) (adaptations to matrices with arbitrary largest eigenvalue will be discussed in the experimental section). To approximate the action of \(\exp(-\tau \mathbf{L})\), where \(\tau \geq 0\), we use the matrix polynomial $p_K(\mathbf{L})$ where $p_K(\lambda)$ is the polynomial obtained by truncating the series~\eqref{E:diff_approx}. The truncation order \(K\) offers a compromise between computational speed and numerical accuracy. The recurrence relations~\eqref{E:recur_cheb} on Chebyshev polynomials yields recurrence relations to compute \(\tilde{T}_k(\mathbf{L})x=T_k(\mathbf{L}-\mathbf{Id})x\).
Given a polynomial order \(K\), computing $p_K(\mathbf{L})x$ requires \(K\) matrix-vector products for the polynomials, and \(K+1\) Bessel function evaluations for the coefficients.
This cost is dominated by the \(K\) matrix-vector products, which can be very efficient if $\mathbf{L}$ is a sparse matrix.

\subsection{Generic bounds on relative approximation errors}

Denote \(p_K\) the polynomial obtained by truncation at order \(K\) of the Chebyshev expansion~\eqref{E:diff_approx}.
For a given input %
vector \(x \neq 0\), one can measure a relative error as:
\begin{equation}
    \label{E:relativeerror1}
    \epsilon_K(x) := \frac{
        \|\exp(-\tau\Lmat)x - p_K(\Lmat)x\|_2^2
    }{
        \|x\|_2^2
    }.
\end{equation}
Expressing \(\exp(-\tau\Lmat)\) and \(p_K(\Lmat)\) in an orthonormal eigenbasis of \(\Lmat\) yields a worst-case relative error:
\begin{equation}
    \label{E:worstcaserelativeerror1}
    \epsilon_K
        := \sup_{x \neq 0} \epsilon_K(x)
        = \max_{i} |h_\tau(\lambda_i)-p_K(\lambda_i)|^2
        \leq \|h_{\tau}- p_{K}\|_{\infty}^2
\end{equation}
with $\lambda_i \in [0,\lambda_{\max}]$ 
the eigenvalues of \(\Lmat\) and \(\|g\|_{\infty} := \sup_{\lambda\in [0, \lambda_{\mathrm{max}}]} |g(\lambda)|\).

\begin{lemma}
    \label{L:uniformbound1}
     Consider $\tau \geq 0$, \(h_\tau\) as in Section~\ref{S:genericsetting}, and \(\Lmat\) a PSD matrix with largest eigenvalue \(\lambda_{\max} = 2\).
    Consider \(p_K\) as above where \(K > \tau/2-1\). With \(C:=\nicefrac{\tau}{2}\) we have
    \begin{equation}
        \|h_{\tau}- p_{K}\|_{\infty}
        \leq
        2 e^{\frac{(\tau/2)^2}{K+2}-\tau} \frac{(\tau/2)^{K+1}}{K! (K+1-\tau/2)} =: g(K,\tau).
    \end{equation}
\end{lemma}

\begin{proof}
Denote $C = \tau/2$. For \(K > C-1\) we have:
\begin{align}
    \label{E:tmp1}
    \sum_{k=K+1}^{\infty} \frac{C^{k}}{k!}
        &\leq \frac{1}{K!}\sum_{k=K+1}^{\infty} \frac{C^{k}}{(K+1)^{k-K}} 
        = \frac{C^{K}}{ K!} \sum_{\ell=1}^{\infty} \frac{C^{\ell}}{(K+1)^{\ell}} \notag\\
        &= \frac{C^{K+1}}{ K!(K+1-C)}
\end{align}
and \(C^2/(K+1)<C\) hence for \(k \geq K+1\)~\eqref{E:dk_bound} yields:
\begin{equation}
    \label{E:dk_bound_bis}
    1 \leq d_k \leq \exp(C^2/(K+2)) \leq \exp(C).
\end{equation}

Since $|T_k(t)| \leq 1$ on $[-1,1]$ (recall that $T_k(\cos\theta) = \cos(k\theta)$), we obtain using Corollary~\ref{L:coeffdkbk}:
\begin{align*}
    \|h_{\tau}- p_{K}\|_{\infty}
        &\stackrel{\eqref{E:diff_approx}}{=} \sup_{\lambda\in [0, \lambda_{\mathrm{max}}]} \left| \sum_{k>K}^{\infty} c_k(\tau) \tilde{T}_k(\lambda) \right| 
        \leq \sum_{k>K}^{\infty}\left|d_k \bar{c}_k \right| \\
        &\stackrel{\eqref{E:cbark},\eqref{E:dk_bound_bis}}{\leq} \exp\left( \tfrac{C^{2}}{K+2}\right) 2 \exp\left(-2C\right) \sum_{k>K}^{\infty} \frac{C^{k}}{k!} \\
        &\stackrel{\eqref{E:tmp1}}{\leq}
        2 \exp\left(\tfrac{C^2}{K+2}-2C\right) 
        \frac{C^{K+1}}{ K!(K+1-C)}. \qedhere
\end{align*}

\end{proof}

While~\eqref{E:relativeerror1} is the error of approximation of \(\exp(-\tau \mathbf{L})x\), \emph{relative to the input 
energy} \(\|x\|_2^2\), an alternative is to measure this error w.r.t. \emph{the output 
energy} \(\|\exp(-\tau\mathbf{L})x\|_2^2\):
\begin{equation}
    \eta_K(x) := \frac{
        \|\exp(-\tau\mathbf{L})x-p_K(\mathbf{L})x\|_2^2
    }{
        \|\exp(-\tau\mathbf{L})x\|_2^2
    }.
\end{equation}

Since \(\|\exp(-\tau \mathbf{L})x\|_2 \geq e^{-\tau\lambda_{\max}} \|x\|_2 = e^{-2\tau}\|x\|_2\) we have \(\eta_K(x) \leq \|h_\tau-p_K\|_\infty^2 e^{4\tau}\).
Using Lemma~\ref{L:uniformbound1} we obtain for \(K>\tau/2-1\) and any \(x\):
\begin{align}
    \epsilon_K(x) & \leq  g^2(K,\tau);\label{E:Bound1}\\
    \eta_K(x) & \leq g^2(K,\tau) e^{4\tau}.\label{E:Bound2}
\end{align}

\subsection{Specific bounds on relative approximation errors}

As the bounds~\eqref{E:Bound1}-\eqref{E:Bound2} are worst-case estimates, they may be improved for a specific input signal \(x\) by taking into account its properties.
To illustrate this, let us focus on graph diffusion where \(\mathbf{L}\) is a graph Laplacian, assuming that \(a_1 := \sum_i x_i \neq 0\).
Since \(a_1/\sqrt{n}\) is the inner product between \(x\) and the unit constant vector \((1,\ldots,1)/\sqrt{n}\), which is an eigenvector of the graph Laplacian \(\mathbf{L}\) associated to the zero eigenvalue \(\lambda_1=0\), we have \(\|\exp(-\tau \mathbf{L})x\|_2^2 \geq |a_1/\sqrt{n}|^2\).
For \(K>\tau/2-1\) this leads to the bound:
\begin{equation}
    \label{E:Bound3}
    \eta_K(x) \leq \epsilon_{K}(x) \frac{\|x\|_2^2}{a_1^2/n} \leq g^2(K,\tau) \frac{n \|x\|_2^2}{a_1^2}.
\end{equation}
This bound improves upon~\eqref{E:Bound2} if
\(e^{4\tau} \geq \frac{n\|x\|_2^2}{a_1^2}\), i.e. when
\begin{equation}
    \label{E:bound_comp_cond}
    \tau \geq \frac{1}{4} \log \frac{n\|x\|_2^2}{a_1^2}.
\end{equation}

\section{Experiments}
\label{S:experiments}

Considering a graph with Laplacian $\Lmat$, the diffusion of a graph signal $x$ at scale $\tau$ is obtained by computing $\exp(-\tau\Lmat)x$. In general, the largest eigenvalue of $\Lmat$ is not necessarily $\lambda_{\max}=2$ (except for example if $\Lmat$ is a so-called \emph{normalized} graph Laplacian, instead of a \emph{combinatorial} graph Laplacian). To handle this case with the polynomial approximations studied in the previous section, we first observe that $\exp(-\tau \Lmat) = \exp(-\tau' \Lmat')$ where $\Lmat' = 2\Lmat/\lambda_{\max}$ and $\tau' = \lambda_{\max}\tau/2$.
Using Equation~\eqref{E:bound_comp_cond} with scale \(\tau'\) allows to select which of the two bounds~\eqref{E:Bound2} or~\eqref{E:Bound3} is the sharpest. The selected bound is then used to find a polynomial order \(K\) that satisfies a given precision criterion.
Then, we can use the recurrence relations~\eqref{E:def_chebycheb_coefficients} to compute the action of the polynomials $\tilde{T}_k(\Lmat')=T_k(\Lmat'-\mathbf{Id})$ on $x$ \cite{DBLP:conf/dcoss/ShumanVF11},
and combine them with the coefficients $c_k(\tau')$ given by~\eqref{E:coeff_as_bessel_fun}.

\subsection{Bound tightness}

Our new bounds accuracy can be illustrated by plotting the minimum truncated order \(K\) required to achieve a given precision. The new bounds can be compared to the tightest bound we could find in the literature~\cite{DBLP:journals/nla/BergamaschiV00}:
\begin{equation}
    \label{berga_ac_dep}
    \eta_K(x) \leq 4 {E(K)}^{2} \frac{n \|x\|_2^2}{a_1^2}
\end{equation}
where \(a_1 = \sum_{i} x_i\), and:
\begin{equation}
    E(K) = \begin{cases}
        e^{\frac{-b(K+1)^{2}}{2\tau}}\left(1+\sqrt{\frac{\pi\tau/2}{b}}\right)+\frac{d^{2\tau}}{1-d} & \text { if } K \leq 2\tau \\
        \frac{d^{K}}{1-d} & \text { if } K >2\tau
    \end{cases}
\end{equation}
with \(b = \frac{2}{1 + \sqrt5}\) and \(d= \frac{\exp(b)}{2 + \sqrt5}\).
This bound can be made independent of \(x\) by using the same procedure as that of used to establish~\eqref{E:Bound2}:
\begin{equation}
    \label{berga_ss_dep}
    \eta_K(x) \leq 4 E(K)^{2} \exp(4\tau).
\end{equation}
\begin{figure}[htbp]
    \centering
    \includegraphics[width=\columnwidth]{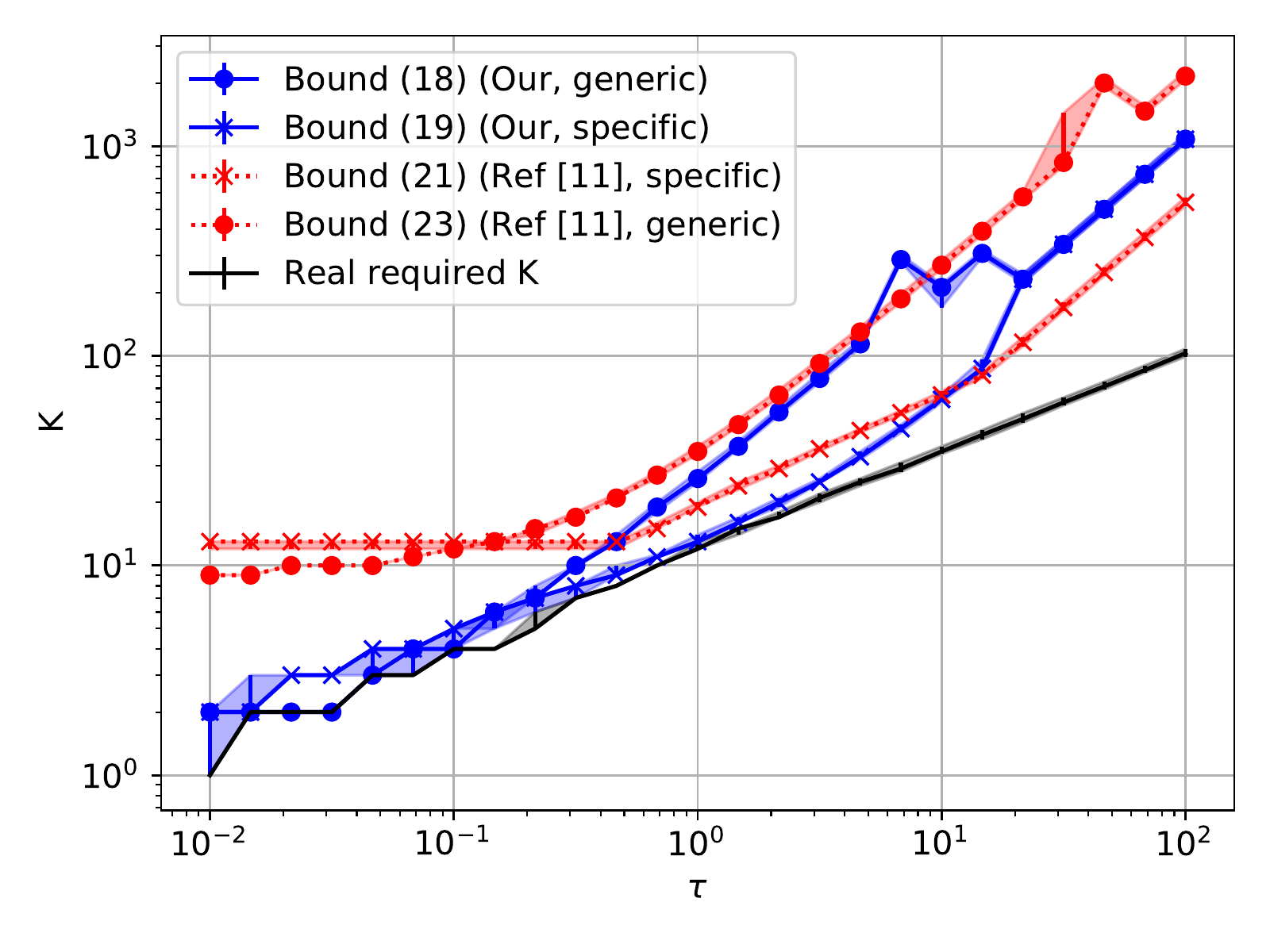}
    \caption{
        Minimum order K to achieve an error \(\eta_K(x)\) below \(10^{-5}\), either real or according to each bound.
        Median values taken for 100 Erdos-Reyni graphs of size 200 with 5\% connection probability, and a centered standard normal distributed signal.}
    \label{F:bound_examination}
\end{figure}
An experiment was performed over 25 values of \(\tau\) ranging from \(10^{-2}\) to \(10^{2}\), 100 samplings of Erdos-Reyni graphs of size \(n=200\), with connection probability \(p=5\%\) (which yields \(\lambda_{max} \simeq 20\)), and coupled with a random signal with entries drawn i.i.d. from a centered standard normal distribution. 
For each set of experiment parameters, for each bound, generically noted $B(K,\tau,x)$, the minimum order \(K\) ensuring \(\eta_K(x) \leq B(K,\tau,x) \leq 10^{-5}\) was computed, as well as the oracle minimum degree \(K\) guaranteeing MSE \(\eta_K(x) \leq 10^{-5}\). The median values over graph generations are plotted on Fig~\ref{F:bound_examination} against \(\tau\), with errorbars using quartiles.

We observe that our new bounds (blue) follow more closely the true minimum \(K\) (black) achieving the targeted precision, up to \(\tau \simeq 10\), thus saving computations over the one of~\cite{DBLP:journals/nla/BergamaschiV00} (red).
Also of interest is the fact that the bounds \eqref{E:Bound3}-\eqref{berga_ac_dep} specific to the input signal are much tighter than their respective generic counterparts \eqref{E:Bound2}-\eqref{berga_ss_dep}.

\subsection{Acceleration of multiscale diffusion}
\label{S:expe_speed_er}

When diffusing at multiple scales \(\{\tau_1\cdots\tau_m\}\), it is worth noting that computations can be factorized.
The order \(K\) can be computed only once (using the largest \(\tau'_i\)), as well as $\tilde{T}_k(\mathbf{L}')x$.
Eventually, the coefficients can be evaluated for all values \(\tau_i\) to generate the needed linear combinations of $\tilde{T}_k(\Lmat')x$, $0 \leq k \leq  K$.
In order to illustrate this speeding-up phenomenon, our method is  compared to \texttt{scipy.\allowbreak{}sparse.\allowbreak{}linalg.\allowbreak{}expm\_multiply}, from the standard \texttt{SciPy} Python package, which uses a Taylor approximation combined with a squaring-and-scaling method.
See~\cite{DBLP:journals/siamsc/Al-MohyH11} for details.

For a first experiment, we take the Standford bunny~\cite{riener2012virtual}, a graph built from a rabbit ceramic scanning (2503 nodes and 65.490 edges, with \(\lambda_{max} \simeq 78\)).
For the signal, we choose a Dirac located at a random node.
We compute repeatedly the diffusion from 2 to 20 scales \(\tau\) sampled in \([10^{-3}, 10^{1}]\).
Our method is set with a target error \(\eta_K \leq 10^{-5}\).
When the \(\tau\) values are linearly spaced, both methods can make use of their respective multiscale acceleration.
In this context, our method is about twice faster than \texttt{Scipy}'s; indeed, it takes 0.36\,s plus 6.1\(\times10^{-3}\)\,s per scale, while \texttt{Scipy}'s takes 0.74\,s plus 2.4\(\times10^{-3}\)\,s per scale.

On the other hand, when the \(\tau\) values are uniformly sampled at random, \texttt{SciPy} cannot make use of its multiscale acceleration.
Indeed, its computation cost 
increases linearly with the number of \(\tau\)'s, with an average cost of 0.39\,s per scale.
Whereas, the additional cost for repeating our method for each new \(\tau\) is negligible (0.0094\,s on average) compared to the necessary time to initialize once and for all, the \(\tilde{T}_k(\mathbf{L}')x\) (0.30\,s).

The trend observed here holds for larger graphs as well.
We run a similar experiment on the \texttt{ogbn-arxiv} graph from the OGB datasets~\cite{hu2021open}.
We take uniformly sampled scales in \([7.6\times 10^{-2}, 2.4 \times 10^{-1}]\) (following recommendations of~\cite{Donnat_2018}), and set our method for \(\eta_K \leq 10^{-3}\).
We observe an average computation time of 504\,s per scale (i.e. 1\,hr and 24\,min for 10 scales) for \texttt{Scipy}'s method, and 87\,s plus 50\,s per scale for our method (i.e. around 9\,min for 10 scales).
If we impose a value $\eta_K \leq 2^{-24}$, comparable to the floating point precision achieved by Scipy, the necessary polynomial order $K$ only increases by 6\%, which does not jeopardise the computational gain of our method.
This behavior gives insight into the advantage of using our fast approximation for addressing the multiscale diffusion on very large graphs.

All experiments are in Python using \texttt{NumPy} and \texttt{SciPy}.
They ran on a Intel-Core i5-5300U CPU with 2.30GHz processor and 15.5 GiB RAM on a Linux Mint 20 Cinnamon.

\section{Conclusion}
\label{S:conclusion}
Our contribution is twofold: first, using the now classical Chebyshev approximation of the exponential function, we significantly improved the state of the art theoretical bound used to determine the minimum polynomial order needed for an expected approximation error. 
Second, in the specific case of the heat diffusion kernel applied to a graph structure, we capitalized on the polynomial properties of the Chebyshev coefficients to factorize the calculus of the diffusion operator, reducing thus drastically its computational cost when applied for several values of the diffusion time. 

The first contribution is particularly important when dealing with the exponential of extremely large matrices, not necessarily coding for a particular graph.
As our new theoretical bound guarantees the same approximation precision for a polynomial order downsized by up to one order of magnitude, the computational gain is considerable when modeling the action of operators on large mesh grids, as it can be the case, for instance, in finite element calculus.  

Our second input is directly related to our initial motivation in \cite{barbe2020graph} that was to identify the best diffusion time $\tau$ in an optimal transport context.
Thanks to our accelerated algorithm, we can afford to repeatedly compute the so-called Diffused Wasserstein distance to find the optimal domain adaptation between graphs' measures. 

\section*{Acknowledgment}
The authors wish to thank Nicolas Brisebarre for discussions that helped sharpening some bounds, as well as Hakim Hadj-Djilani for discussions on python implementations.

\bibliographystyle{IEEEtran}

\begin{thebibliography}{10}
\providecommand{\url}[1]{#1}
\csname url@samestyle\endcsname
\providecommand{\newblock}{\relax}
\providecommand{\bibinfo}[2]{#2}
\providecommand{\BIBentrySTDinterwordspacing}{\spaceskip=0pt\relax}
\providecommand{\BIBentryALTinterwordstretchfactor}{4}
\providecommand{\BIBentryALTinterwordspacing}{\spaceskip=\fontdimen2\font plus
\BIBentryALTinterwordstretchfactor\fontdimen3\font minus
  \fontdimen4\font\relax}
\providecommand{\BIBforeignlanguage}[2]{{%
\expandafter\ifx\csname l@#1\endcsname\relax
\typeout{** WARNING: IEEEtran.bst: No hyphenation pattern has been}%
\typeout{** loaded for the language `#1'. Using the pattern for}%
\typeout{** the default language instead.}%
\else
\language=\csname l@#1\endcsname
\fi
#2}}
\providecommand{\BIBdecl}{\relax}
\BIBdecl

\bibitem{mattheij2005partial}
R.~M. Mattheij, S.~W. Rienstra, and J.~T.~T. Boonkkamp, \emph{Partial
  differential equations: modeling, analysis, computation}.\hskip 1em plus
  0.5em minus 0.4em\relax SIAM, 2005.

\bibitem{de2019analysis}
O.~De~la Cruz~Cabrera, M.~Matar, and L.~Reichel, ``Analysis of directed
  networks via the matrix exponential,'' \emph{Journal of Computational and
  Applied Mathematics}, vol. 355, pp. 182--192, 2019.

\bibitem{zhuang2013power}
H.~Zhuang, S.-H. Weng, and C.-K. Cheng, ``Power grid simulation using matrix
  exponential method with rational krylov subspaces,'' in \emph{2013 IEEE 10th
  International Conference on ASIC}.\hskip 1em plus 0.5em minus 0.4em\relax
  IEEE, 2013, pp. 1--4.

\bibitem{pusa2010computing}
M.~Pusa and J.~Lepp{\"a}nen, ``Computing the matrix exponential in burnup
  calculations,'' \emph{Nuclear science and engineering}, vol. 164, no.~2, pp.
  140--150, 2010.

\bibitem{barbe2020graph}
A.~Barbe, M.~Sebban, P.~Gon{\c{c}}alves, P.~Borgnat, and R.~Gribonval, ``Graph
  diffusion wasserstein distances,'' in \emph{European Conference on Machine
  Learning and Principles and Practice of Knowledge Discovery in Databases},
  2020.

\bibitem{chung1997spectral}
F.~R. Chung and F.~C. Graham, \emph{Spectral graph theory}.\hskip 1em plus
  0.5em minus 0.4em\relax American Mathematical Soc., 1997, no.~92.

\bibitem{DBLP:journals/adcm/MehraSL21}
\BIBentryALTinterwordspacing
M.~Mehra, A.~Shukla, and G.~Leugering, ``An adaptive spectral graph wavelet
  method for pdes on networks,'' \emph{Adv. Comput. Math.}, vol.~47, no.~1,
  p.~12, 2021. [Online]. Available:
  \url{https://doi.org/10.1007/s10444-020-09824-9}
\BIBentrySTDinterwordspacing

\bibitem{DBLP:journals/siammax/PopolizioS08}
\BIBentryALTinterwordspacing
M.~Popolizio and V.~Simoncini, ``Acceleration techniques for approximating the
  matrix exponential operator,'' \emph{{SIAM} J. Matrix Anal. Appl.}, vol.~30,
  no.~2, pp. 657--683, 2008. [Online]. Available:
  \url{https://doi.org/10.1137/060672856}
\BIBentrySTDinterwordspacing

\bibitem{DBLP:journals/jcam/BotchevK20}
\BIBentryALTinterwordspacing
M.~A. Botchev and L.~A. Knizhnerman, ``{ART:} adaptive residual-time restarting
  for krylov subspace matrix exponential evaluations,'' \emph{J. Comput. Appl.
  Math.}, vol. 364, 2020. [Online]. Available:
  \url{https://doi.org/10.1016/j.cam.2019.06.027}
\BIBentrySTDinterwordspacing

\bibitem{DRUSKIN1989112}
\BIBentryALTinterwordspacing
V.~Druskin and L.~Knizhnerman, ``Two polynomial methods of calculating
  functions of symmetric matrices,'' \emph{USSR Computational Mathematics and
  Mathematical Physics}, vol.~29, no.~6, pp. 112--121, 1989. [Online].
  Available:
  \url{https://www.sciencedirect.com/science/article/pii/S0041555389800205}
\BIBentrySTDinterwordspacing

\bibitem{DBLP:journals/nla/BergamaschiV00}
\BIBentryALTinterwordspacing
L.~Bergamaschi and M.~Vianello, ``Efficient computation of the exponential
  operator for large, sparse, symmetric matrices,'' \emph{Numer. Linear Algebra
  Appl.}, vol.~7, no.~1, pp. 27--45, 2000. [Online]. Available:
  \url{https://doi.org/10.1002/(SICI)1099-1506(200001/02)7:1\<27::AID-NLA185\>3.0.CO;2-4}
\BIBentrySTDinterwordspacing

\bibitem{mason2002chebyshev}
J.~C. Mason and D.~C. Handscomb, \emph{Chebyshev polynomials}.\hskip 1em plus
  0.5em minus 0.4em\relax CRC press, 2002.

\bibitem{DBLP:journals/siamsc/Al-MohyH11}
A.~H. Al{-}Mohy and N.~J. Higham, ``Computing the action of the matrix
  exponential, with an application to exponential integrators,'' \emph{{SIAM}
  J. Scientific Computing}, vol.~33, no.~2, pp. 488--511, 2011.

\bibitem{phillips2003interpolation}
\BIBentryALTinterwordspacing
G.~Phillips, \emph{Interpolation and Approximation by Polynomials}, ser. CMS
  Books in Mathematics.\hskip 1em plus 0.5em minus 0.4em\relax Springer, 2003.
  [Online]. Available: \url{https://books.google.fr/books?id=87vciTxMcF8C}
\BIBentrySTDinterwordspacing

\bibitem{abramowitz1988handbook}
\BIBentryALTinterwordspacing
M.~Abramowitz, I.~A. Stegun, and R.~H. Romer, ``Handbook of mathematical
  functions with formulas, graphs, and mathematical tables,'' 1988. [Online].
  Available: \url{http://www.math.ubc.ca/~cbm/aands/toc.htm}
\BIBentrySTDinterwordspacing

\bibitem{DBLP:conf/dcoss/ShumanVF11}
\BIBentryALTinterwordspacing
D.~I. Shuman, P.~Vandergheynst, and P.~Frossard, ``Chebyshev polynomial
  approximation for distributed signal processing,'' in \emph{Distributed
  Computing in Sensor Systems, 7th {IEEE} International Conference and
  Workshops, {DCOSS} 2011, Barcelona, Spain, 27-29 June, 2011,
  Proceedings}.\hskip 1em plus 0.5em minus 0.4em\relax {IEEE} Computer Society,
  2011, pp. 1--8. [Online]. Available:
  \url{https://doi.org/10.1109/DCOSS.2011.5982158}
\BIBentrySTDinterwordspacing

\bibitem{DBLP:journals/actanum/HighamA10}
\BIBentryALTinterwordspacing
N.~J. Higham and A.~H. Al{-}Mohy, ``Computing matrix functions,'' \emph{Acta
  Numer.}, vol.~19, pp. 159--208, 2010. [Online]. Available:
  \url{https://doi.org/10.1017/S0962492910000036}
\BIBentrySTDinterwordspacing

\bibitem{riener2012virtual}
R.~Riener and M.~Harders, \emph{Virtual reality in medicine}.\hskip 1em plus
  0.5em minus 0.4em\relax Springer Science \& Business Media, 2012.

\bibitem{hu2021open}
W.~Hu, M.~Fey, M.~Zitnik, Y.~Dong, H.~Ren, B.~Liu, M.~Catasta, and J.~Leskovec,
  ``Open graph benchmark: Datasets for machine learning on graphs,'' 2021.

\bibitem{Donnat_2018}
\BIBentryALTinterwordspacing
C.~Donnat, M.~Zitnik, D.~Hallac, and J.~Leskovec, ``Learning structural node
  embeddings via diffusion wavelets,'' \emph{Proceedings of the 24th ACM SIGKDD
  International Conference on Knowledge Discovery \& Data Mining}, Jul 2018.
  [Online]. Available: \url{http://dx.doi.org/10.1145/3219819.3220025}
\BIBentrySTDinterwordspacing

\end{thebibliography}


\end{document}